\documentclass[11pt]{article}

\usepackage[margin=1in]{geometry}

\usepackage{amsmath}
\usepackage{amssymb}
\usepackage{fullpage}
\usepackage{color}
\usepackage{graphicx}
\usepackage{epsfig}
\usepackage{amsthm}
\usepackage{latexsym}
\usepackage{verbatim}
\usepackage{cite}
\usepackage{url}
\usepackage{tabularx}
\usepackage{algorithmic}
\usepackage{algorithm}
\usepackage{subfig}
\usepackage{comment}

\newtheorem*{mcd}{McDiarmid Inequality}{\bfseries}{\itshape}

\newtheorem{theorem}{Theorem}

\newtheorem{lemma}{Lemma}
\newtheorem{definition}{Definition}

\newcommand{\De}{\mathcal{D}}
\newcommand{\Ch}{\mathcal{C}}

\newcommand{\E}{\mathbb{E}}

\begin{document}
\title{Approximating Nash Equilibria in \\ Tree Polymatrix Games}

\author{Siddharth Barman\thanks{Indian Institute of Science. \tt{barman@csa.iisc.ernet.in}} \and Katrina Ligett\thanks{California Institute of Technology and The Hebrew University of Jerusalem. \tt{katrina@caltech.edu}} \and Georgios Piliouras\thanks{Singapore University of Technology and Design. \tt{georgios@sutd.edu.sg}}}

\date{}

\maketitle

\abstract{
We develop a quasi-polynomial time Las Vegas algorithm for approximating Nash equilibria in polymatrix games over trees, under a mild renormalizing assumption. Our result, in particular, leads to an expected polynomial-time algorithm for computing approximate Nash equilibria of tree polymatrix games in which the number of actions per player is a fixed constant. Further, for trees with constant degree, the running time of the algorithm matches the best known upper bound for approximating Nash equilibria in bimatrix games (Lipton, Markakis, and Mehta 2003). 

Notably, this work closely complements the hardness result of Rubinstein (2015), which establishes the inapproximability of Nash equilibria in polymatrix games over constant-degree bipartite graphs with two actions per player.
}

\section{Introduction}

The complexity of equilibrium computation is a central area of research in algorithmic game theory. Recent years have seen significant progress in this line of work, especially in the context of two-player games~\cite{DGP, chen2009settling, daskalakis2013complexity, AR, LMM, daskalakis2006note, DMPprogress, tsaknakis2007optimization}. Furthermore, the computation of approximate Nash equilibrium in games over networks has emerged as an important research direction~\cite{kearns2001graphical, ortiz2002nash, elkind2006nash, daskalakis2009network, cai2011minmax}. Motivation for studying such multiplayer games stems in part from the prevalence and importance of large networks of interconnected, self-interested agents. 

The prototypical family of large network games is that of \emph{polymatrix games}. These games merge two classical concepts, two-player games and networks. In a polymatrix game, each player corresponds to a node in a network, and each edge encodes a two-player game between the two endpoints of the edge. A player's payoff is the sum of her payoffs across the bimatrix games (edges) she participates in.
Polymatrix games capture complex settings with arbitrarily many players while keeping the description complexity of the game polynomially small in the number of players. Computation of equilibria for polymatrix games is hence a natural test case, and has emerged at the boundary of computational tractability.

The seminal {\rm PPAD} hardness reductions for computing  $\varepsilon$-Nash equilibria\footnote{In an $\varepsilon$-Nash equilibrium, a player can gain at most $\varepsilon$ by unilaterally deviating from her current strategy.} by Daskalakis, Goldberg, and Papadimitriou \cite{DGP} along with their extensions by Chen, Deng, and Teng \cite{chen2009settling}  to two-player games were crucially developed within the context of polymatrix games.\footnote{These hardness result hold for polynomially small  $\varepsilon$.} Recently, Rubinstein \cite{AR}  strengthened these inapproximability guarantees by establishing that there exists a constant $\varepsilon$ such that finding an $\varepsilon$-Nash equilibrium in polymatrix games over bipartite graphs of constant degree is computationally hard. Our positive algorithmic result is inspired by this work and {\em explores the boundary between tractability and intractability of $\varepsilon$-Nash computation in polymatrix games.} \\

The study of equilibria in polymatrix games has had a long history \cite{Yanovskaya,howson1972equilibria}. To avoid  hardness, most algorithmic results have focused on structured subclasses of polymatrix games. These include polymatrix generalizations of zero-sum games where exact Nash equilibria can be computed in polynomial time \cite{daskalakis2009network,cai2011minmax}. Games on trees is another family of multiplayer games that has received attention~\cite{kearns2001graphical, ortiz2002nash, elkind2006nash}. The proposed algorithm in \cite{elkind2006nash} finds an exact Nash equilibrium in two-action games on paths and runs in polynomial time, but in the case of trees the running time may be exponential even if the degree of the underlying tree is bounded. In contrast, we study computation of approximate Nash equilibrium in trees of arbitrary degree, and develop an algorithm that runs in quasi-polynomial time. Finally, some interesting progress has been made in the case of general polymatrix games as well, where it has been shown that a $(0.5+\varepsilon)$-Nash equilibrium of a polymatrix game can be computed in time polynomial in the input size and $1/\varepsilon^2$ \cite{DFSS-WINE}.

 \paragraph{\bf Results.} 
 We develop a quasi-polynomial time algorithm for approximating Nash equilibrium in polymatrix games over trees under a mild renormalizing assumption on the players' payoffs. Specifically, instead of normalizing the entries of each bimatrix game to lie in $[0,1]$, which results in each player $i$'s payoff depending linearly on its degree, we normalize them to lie in $[0,1/\text{degree}(i)]$, so that players' total payoffs lie in $[0,1]$. Our results actually extend even under weaker renormalization conditions; see Section~\ref{sec:notation} for details. We show that, given an $n$-player, $m$-action normalized polymatrix game over a tree, we can find an $\varepsilon$-Nash equilibrium of the game in expected time \[ m^{O\left( \frac{ \log m (\log m + \log n - \log \varepsilon) }{\varepsilon^4} \right)}. \] 

Our approach immediately implies a polynomial time approximation scheme for computing Nash equilibria when the number of actions per player is  constant. The case of standard bimatrix games can be trivially captured in our setting via a single-edge polymatrix game. Further, for trees of constant degree our  framework yields an algorithm that finds an $\varepsilon$-Nash equilibrium in time $m^{O\left(\frac{\log m + \log n}{\varepsilon^2}\right)}$. Note that in the single edge case (i.e., the case of standard bimatrix games) this running-time bound matches the best known upper bound for approximating Nash equilibria~\cite{LMM}. 

\noindent    
{\bf Techniques.} We develop a dynamic program to find an approximate Nash equilibrium of the given tree polymatrix game. The idea is to root the underlying tree and process it in a bottom-up manner. For each node/player $p$ we maintain a set of mixed strategies---i.e., probability distributions over player's actions---that can be extended into a ``partial'' (approximate) equilibrium of the subtree rooted at the node. That is, for each mixed strategy assigned to $p$ there exist mixed strategies for the descendants of $p$ under which no descendant can benefit more than $\varepsilon$, in expectation, by unilateral deviation. We find such extendable mixed strategies of a player $p$ after processing all of its children; in other words, we start from the leaves of the tree and move towards the root. Note that such an extendable mixed strategy for the root corresponds to an approximate Nash equilibrium of the game. Also, it is worth pointing out that the tree structure enables us to find partial equilibria of disjoint subtrees separately. In particular, the fact that the utilities of players depend only on the actions of its parent and its children implies that disjoint subtrees can be processed separately.  

In and of itself, using a dynamic program to find an approximate Nash equilibrium over a tree is a natural idea. In fact, similar approaches have been adopted in prior work; see, e.g.,~\cite{elkind2006nash}. The key technical contribution in this paper is to show that the update step in the dynamic program can be performed in quasi-polynomial time. To do this, we focus on a specific set of mixed strategies $U$, which is the set of all uniform distributions with support size polynomial in the approximation parameter $\varepsilon$ and logarithmic in the number of players and the number of actions; see Section~\ref{sec:notation} for a formal definition. It was shown in~\cite{BBP} that every multiplayer game admits an $\varepsilon$-Nash equilibrium wherein the mixed strategy of each player is contained in $U$. Hence, given an $n$-player game, an exhaustive search over the set $U^n$ is guaranteed to find an approximate Nash equilibrium. But, such a search runs in time exponential in $n$.  We show that for tree polymatrix games, an exponential-time exhaustive search can be bypassed. The idea is to follow the above mentioned dynamic program and consider, for each player $p$, mixed strategies in the set $U$  that can be extended into partial equilibria of the subtree rooted at $p$. To perform the update step in the dynamic program we employ a linear program that, interestingly, gives a tight characterization of mixed strategies that can be extended. Together, these ideas lead us to a quasi-polynomial time approximation algorithm.

\section{Notation and Preliminaries} 
\label{sec:notation}

We study games with $n$ players and $m$ actions per player.\footnote{We assume that each player has $m$ actions for ease of presentation. The developed result directly extends to the case wherein the number of actions of each player is different.} Write $[n]$ and $[m]$ to denote the set of players and the set of actions of each player, respectively. The utilities of the players are normalized between $0$ and $1$; in particular, for each player $p$ we have utility  $u_p: [m]^n \rightarrow [0,1]$. Let $\Delta^m$ be the set of probability distributions over $[m]$. In addition, for mixed strategy profile $x = (x_q)_{q \in [n]} \in \Delta^m \times \ldots \times \Delta^m$, we denote the expected utility of player $p$ by $u_p(x)$. Following standard notation, we use $x_{-p}$ to denote the mixed strategy profile of all players besides $p$.   

\begin{definition}[$\varepsilon$-Nash equilibirum]
A mixed strategy profile $x = (x_q)_{q \in [n]}$, where each $x_q \in \Delta^m$, is said to be an $\varepsilon$-Nash equilibrium iff for every player $p \in [n]$ and action $a \in [m]$ we have $u_p(x) \geq u_p(a, x_{-p}) - \varepsilon$.
\end{definition}

Here, setting $\varepsilon = 0$ gives us the definition of a Nash equilibrium. 

\paragraph{\bf Polymatrix Games.}
In a polymatrix game, the players correspond to vertices of a graph $G=(V,E)$ and the utility of each player $p \in V$ depends only on her action and the actions of her neighbors. Moreover, the utility of each player is {\em separable}, i.e., for each edge $(p,q) \in E$ we have a bimatrix game specified by  $m \times m$ matrices $A_{p,q}$ and $A_{q,p}$, and the utility of player $p$, under action profile $(a_q)_{q \in [n]} \in [m]^n$, is specified as follows $u_p(a_1, a_2,\ldots, a_n) := \sum_{q : (p,q)\in E }\  e_{a_p}^T \ A_{p,q} \ e_{a_q}$. Here, $e_k \in \mathbb{R}^m$ denotes the standard basis vector with $1$ in the $k$th component and $0$'s elsewhere. Along these lines, for a mixed strategy profile $(x_q)_{q \in [n]} \in \left(\Delta^m\right)^n$, the expected utility of player $p$, $u_p(x_1, x_2, \ldots, x_n) := \sum_{q : (p,q)\in E} x_p^T A_{p,q} x_q$.

As mentioned above, the utility of each player is normalized between $0$ and $1$. A typical way to accomplish this normalization (see, e.g.,~\cite{DFSS-WINE}) is to assume that for each player $p \in [n]$ the associated payoff matrices, $A_{p,q}$s, are entry-wise between $0$ and $1$, and the utility of player $p$ with degree $d$ (in the graph) is obtained by dividing the sum of the payoffs by $d$, i.e., $u_p(a_1, a_2,\ldots, a_n) := \frac{1}{d} \sum_{q : (p,q)\in E }\  e_{a_p}^T \ A_{p,q} \ e_{a_q}$. This normalization ensures that the same approximation guarantee is achieved for all players, irrespective of their degrees. If, instead, one assumes that entry-wise the $A_{p,q}$s are between $0$ and $1$ and simply add the payoffs $e_{a_i}^T  A_{p,q} e_{a_j}$, then the approximation guarantee for players with higher degree---since $\varepsilon$ is the same of all the players---is stronger. This would lead to an undesirable,  nonuniform approximation bound. 

The degree-normalized scaling mentioned above is equivalent to the assumption that for player $p \in [n]$, with degree $d$, the matrices $A_{p,q}$s are contained in $[0,1/d]^{m \times m}$ and $u_p(a_1,\ldots, a_n) := \sum_{q : (p,q)\in E }\  e_{a_p}^T \ A_{p,q} \ e_{a_q}$. In this paper we in fact consider a more general setup in which, for a player with degree $d$, entries of $A_{p,q}$s are between $0$ and $\max\left\{\frac{1}{d}, \frac{\varepsilon}{2\sqrt{6 d \log m}}\right\}$. Here, again we assume that for each action profile $a$ we have $u_i(a) \in [0,1]$.  Developing a quasi-polynomial time algorithm without an entry-wise assumption (i.e., without requirement (i) in the following definition) remains an interesting direction for future work. 
 
\begin{definition}[Normalized Polymatrix Game]
\label{def:norm}
Let $\mathcal{G}$ be an $n$-player $m$-action polymatrix game over graph $G=(V,E)$ and with payoff matrices $A_{p,q}$ and $A_{q,p}$, for $(p,q) \in E$. Given parameter $\varepsilon$, we say that $\mathcal{G}$ is normalized iff for each player $p \in [n]$ we have (i) the entries of $A_{p,q}$s are contained in $\left[0, \max\left\{\frac{1}{d}, \frac{\varepsilon}{2\sqrt{ 6 d \log m}}\right\}\right]$; here $d$ is the degree of player $p$ in $G$, and (ii) for every action profile $(a_1, \ldots, a_n) \in [m]^n$, the utility $u_p(a_1, \ldots, a_n) := \sum_{q : (p,q) \in E }\  e_{a_p}^T \ A_{p,q} \ e_{a_q}$ is between $0$ and $1$. 
\end{definition}

Given mixed strategies of the neighbors of a player $p$, say $(x_q)_{q : (p,q) \in E}$, $x_p \in \Delta^m$ is said to be an {\em $\varepsilon$-best response} of $p$ against $(x_q)_{q : (p,q) \in E}$ if $p$ cannot benefit more than $\varepsilon$ in expectation by deviating from $x_p$, i.e.,   
\begin{align}
\label{ineq:br}
\sum_{q : (p,q) \in E} x_p^T A_{p,q} x_q  & \geq  \max_{j \in [m] } \left( \sum_{q : (p,q) \in E}  e_j^T A_{p,q} \ x_q \right) - \varepsilon.
\end{align}

This paper studies polymatrix games in which the underlying graph $G$ is a tree. Note that a polymatrix game with exactly two players over a single edge $(p,q)\in E$---which is trivially a tree---corresponds to a bimatrix game between players $p$ and $q$. Hence, computation of an approximate Nash equilibrium in tree polymatrix games is at least as hard as computation of approximate Nash equilibrium in bimatrix games. Therefore, our running-time benchmark for finding an $\varepsilon$-Nash equilibrium is quasi-polynomial: $m^{O\left(\frac{\log m}{\varepsilon^2}\right)}$, which is the best known upper bound for approximating Nash equilibria in bimatrix games~\cite{LMM}. 

\paragraph{\bf Uniform Probability Distributions.}
A probability distribution $x \in \Delta^m$ is said to be $b$ uniform if it is a uniform distribution over a size-$b$ multiset of $[m]$. Write $U \subset \Delta^m$ to denote the set of all $\left(\frac{8\left(\ln m + \ln n - \ln \varepsilon + \ln 8 \right)}{\varepsilon^2}\right)$-uniform probability distributions. Note that 
\begin{align}
\label{ineq:usize}
|U| & = m^{O\left(\frac{\log m + \log n - \log \varepsilon}{\varepsilon^2}\right)}
\end{align}

As mentioned above, the work of Babichenko et al.~\cite{BBP} establishes that every $n$-player $m$-action game admits an $\varepsilon$-Nash equilibrium $x=(x_q)_{q \in [n]}$ such that $x_q \in U$ for all $q \in [n]$. Hence, an exhaustive search over the set $U^n$ is guaranteed to find an $\varepsilon$-Nash equilibrium. Note that the running time of such a search is $m^{O\left(\frac{n \left(\log m + \log n - \log \varepsilon \right)} {\varepsilon^2}\right)}$, which is exponential in $n$. In contrast to this exponential-time algorithm, we show that for tree polymatrix games an approximate Nash equilibrium can be computed in expected time $m^{O\left( \frac{ \log m (\log m + \log n - \log \varepsilon) }{\varepsilon^4} \right)}$, which is quasi-polynomial in $n$ and $m$.

\hfill \\
Next we state McDiarmid's inequality~\cite{mcd}. We use this concentration bound to prove our main result.
\begin{mcd}
\label{thm:mcd}
Let $Z_1, Z_2, \ldots, Z_d \in \mathcal{Z}$ be independent random variables and  $f: \mathcal{Z}^d \rightarrow \mathbb{R}$ be a function of $Z_1, Z_2, \ldots, Z_d$. If for all $i \in [d]$ and for all $z_1, z_2, \ldots, z_d, z_i' \in \mathcal{Z}$ the function $f$ satisfies
\begin{align*}
| f(z_1, \ldots, z_i, \ldots, z_d) - f(z_i, \ldots, z_i', \ldots, z_d) | & \leq c_i,
\end{align*}
then for $\delta >0$,
\begin{align*}
\Pr( |f - \E[f]| \geq \delta ) & \leq 2 \ \textrm{exp} \left(\frac{-2 \delta^2}{\sum_{i=1}^d c_i^2} \right).
\end{align*}
\end{mcd}

\section{Quasi-Polynomial Time Algorithm}
This section develops the dynamic program that finds an approximate Nash equilibrium. We will consider $G$ to be a rooted tree and process it in a bottom-up manner. We start with players all whose descendants are leaves, and then iteratively proceed onto the remaining players. 

Write $\Ch(q)$ and $\De(q)$ to denote the set of children and the set of descendants of player $q$, respectively. The iterative process maintains a set $U_{p,q}(z)$ for each parent-child pair $(p,q) \in E$ and each $z \in U$.\footnote{Recall that $U$ is the set of all $O\left(\frac{\log m + \log n - \log \varepsilon}{\varepsilon^2}\right)$-uniform probability distributions.} Intuitively, $U_{p,q}(z)$ denotes the set of mixed strategies for player $q$ that can be extended into a ``partial'' $\varepsilon$-Nash equilibrium of the subtree rooted at $q$. Here $p$, the parent of $q$, is playing mixed strategy $z$ and might not be best responding. Formally, the inductive definition of the sets $U_{p,q}(z)$s is as follows:
\begin{itemize}
\item  If $q$ is a leaf player (i.e., $q$ corresponds to a leaf in tree $G$), then $U_{p,q} (z) :=\{ y \in U \mid y$ is an $\varepsilon$-best response of $q$ against $z\} $. 
\item Else, if $q$ is a not a leaf player, we define $U_{p,q} (z) :=\{ y \in U \mid $ there exist mixed strategies $(x_c)_{c \in \Ch(q)} \in \prod_{c \in \Ch(q)} U_{q,c}(y)  $ such that $y$ is an $\varepsilon$-best response of $q$ against $(x_c)_{c \in \Ch(q)}$ and $z \} $; here, mixed strategy $z$ is associated with parent player $p$. 
\end{itemize}

We also define the set $U_r$ for the root $r$ of tree $G$: $U_r := \{ y \in U \mid$ there exist mixed strategies $(x_c)_{c \in \Ch(r)} \in \prod_{c \in \Ch(r)} U_{r,c}(y)  $ such that $y$ is an $\varepsilon$-best response of $r$ against $(x_c)_{c \in \Ch(r)} \}$. 

If $y \in U_{p,q}(z)$ and $q$ is not a leaf, then, by the above definition, there exist mixed strategy profiles $(x_c)_{c \in \Ch(q)} \in \prod_{c \in \Ch(q)} U_{q,c}(y)$ such that $y$ is an $\varepsilon$-best response of $q$ against $(x_c)_{c \in \Ch(q)}$ and $z$. We will use $E_{p,q}(z,y)$ to denote such a collection of mixed strategies, $(x_c)_{c \in \Ch(q)}$. 

Along these lines, for the root $r$ of the tree $G$ we define $E_r(y)$, for each $y \in U_r$, to be a collection of mixed strategies $(x_c)_{c \in \Ch(r)} \in \prod_{c \in \Ch(r)} U_{r,c}(y)  $ such that $y$ is an $\varepsilon$-best response of $r$ against $(x_c)_{c \in \Ch(r)} $.

Note that mixed strategies in $E_{p,q}(z,y)$ extend $y$ into a ``partial'' $\varepsilon$-Nash equilibrium of the subtree rooted at $q$. Specifically, we can inductively use $E_{p,q}(z,y)$, then $E_{q,c}(y, x_c)$, for each $c \in \Ch(q)$, and so on, to determine mixed strategies $(x_s)_{s \in \De(q)}$ for each descendant $s \in \De(q)$ such that no player in the subtree rooted at $q$ can benefit more than $\varepsilon$, in expectation, by deviating unilaterally. Here we do not assert that the parent player $p$ is at an approximate equilibrium. In addition, note that the utilities of all the players $s \in \De(q) \cup \{ q \}$ depend only on the mixed strategies of players in $\De(q) \cup \{ p,q \}$ and, hence, these utilities can be determined even if the mixed strategies of players in $[n]\setminus \left( \De(q) \cup \{ p,q \} \right)$ are unspecified.

{
\begin{algorithm}
{{\bf Given:} A normalized polymatrix game over tree $G=(V,E)$. {\bf Return:} An $\varepsilon$-Nash equilibrium of the game.}
\caption{Algorithm for finding $\varepsilon$-Nash equilibrium in tree polymatrix games}
\label{alg:dp}
  \begin{algorithmic}[1]
   \STATE Initialize processed set $P$ to be the leaves in $G$ and all $U_{p,q}(z)  =  \varnothing$ 
\WHILE{$V \setminus P \neq \phi$}
\STATE{Select $p \in V \setminus P$ such that $\Ch(p) \subseteq P$}
\FORALL{$q \in \Ch(p)$ and $z \in U$}
\FORALL{$y \in U$}
\IF{$q$ is a leaf node and $y$ is an $\varepsilon$-best response of $q$ against $z$}
\STATE Update $U_{p,q}(z) \leftarrow U_{p,q}(z) \cup \{ y \}$
\ELSIF{there exist mixed strategy profiles $(x_c)_{c \in \Ch(q)} \in \prod_{c \in \Ch(q)} U_{q,c}(y)$ such that $y$ is an $\varepsilon$-best response against $(x_c)_{c \in \Ch(q)}$ and $z$} \label{step:test}
\STATE Update $U_{p,q}(z) \leftarrow U_{p,q}(z) \cup \{ y \}$ and set $E_{p,q}(z,y) \leftarrow (x_c)_{c \in \Ch(q)}$  \COMMENT{There could be multiple tuples  $(x_c)_{c \in \Ch(q)}$ that satisfy this best-response condition. We set $E_{p,q}(z,y)$ to be any one of them.}
\ENDIF
\ENDFOR
\ENDFOR
\STATE $P \leftarrow P \cup \{ p \} $
\ENDWHILE     
\STATE For the root $r$ of the tree $G$, initialize $U_r = \phi$.
\FORALL{$y \in U$}
\IF{there exist mixed strategy profiles $(x_c)_{c \in \Ch(r)} \in \prod_{c \in \Ch(r)} U_{r,c}(y)$ such that $y$ is an $\varepsilon$-best response against $(x_c)_{c \in \Ch(r)}$} \label{step:final}
\STATE Update $U_r \leftarrow U_r \cup \{ y\}$ and set $E_r(y) = (x_c)_{c \in \Ch(r)} $. Use Lemma~\ref{lem:backtrack} to find an $\varepsilon$-Nash equilibrium of the game
\ENDIF 
\ENDFOR
\end{algorithmic}
\end{algorithm}
}

Following the definition of $U_{p,q}(z)$, Algorithm~\ref{alg:dp} constructs these sets and extensions $E_{p,q}(z,y)$ for all parent-child pairs $(p,q) \in E$ and $z \in U$ in a bottom-up manner. At the end, the algorithm uses the set $U_{r}$ defined for the root $r$ to find an $\varepsilon$-Nash equilibrium of the game. Overall, the applicability of the sets $U_{p,q}$ and $E_{p,q}$ is established in Lemma~\ref{lem:backtrack} below.  \\

\begin{lemma}
\label{lem:backtrack}
Let $\mathcal{G}$ be a polymatrix game over a tree $G=(V,E)$. Given sets $U_{p,q}(z)$---for each parent-child pair $(p,q) \in E$---and mixed strategy collections $E_{p,q}(z,y)$---for $y \in U_{p,q}(z)$---along with a mixed strategy profile $\hat{y} \in U_r$ and associated collection $E_r(\hat{y})$ for the root $r$, we can find an $\varepsilon$-Nash equilibrium of the game $\mathcal{G}$ in time polynomial in $|U|$. 
\end{lemma}

\begin{proof}
The lemma is implied  directly by the underlying definitions. For each parent-child pair $(p,q) \in E$ there exists at least one set $U_{p,q}$ which is nonempty: as mentioned above, every $n$-player $m$-action game admits an $\varepsilon$-Nash equilibrium $(\hat{x}_q)_{q \in [n]} $ where each $\hat{x}_q \in U$. Hence, in particular, $ \hat{x}_q \in U_{p,q} (\hat{x}_p)$. Moreover, we have $ \hat{x}_r \in U_r$. 

In fact to find an $\varepsilon$-Nash equilibrium we can start with the given mixed strategy profile $x_r = \hat{y}$ then, for each $ c \in \Ch(r)$, pick the corresponding mixed strategy $x_c$ in $E_r(x_r)$. The definition of $E_r(x_r)$ implies that $x_c$ can be extended to obtain an $\varepsilon$-Nash equilibrium of the subtree rooted at $c$. We can in fact find such an  $\varepsilon$-Nash equilibrium by proceeding inductively down the tree; in particular, by setting $x_s$ for $s \in \Ch(c)$ to be the strategy associated with $s$ in $E_{r,c}(x_r, x_c)$). The definitions of $E_{p,q}$s ensure that this inductive process will run to completion and find an $\varepsilon$-Nash equilibrium of the subtree rooted at $c$. By repeating the process for each  $ c \in \Ch(r)$ we will find a mixed strategy $x_p$ for each player $p \in [n]$. Furthermore, the definitions of the underlying sets also imply that the found mixed strategy profile $(x_p)_{p \in [n]}$ is an $\varepsilon$-Nash equilibrium.   
\end{proof}

Algorithm~\ref{alg:dp} tests whether $y \in U_{p,q}(z)$ (i.e., tests whether there exist mixed strategy profiles $(x_c)_{c \in \Ch(q)} \in \prod_{c \in \Ch(q)} U_{q,c}(y)$ such that $y$ is an $\varepsilon$-best response of $q$ against $(x_c)_{c \in \Ch(q)}$ and $z$) in Step~\ref{step:test}, and the same idea is employed in Step~\ref{step:final}.  In particular, if the number of children of $q$ is $\Omega \left( \frac{\log m}{\varepsilon^2 }\right)$ then Algorithm~\ref{alg:dp} uses the following linear-programming relaxation LP($p,q,z,y$) to perform this test. The other case, wherein $|\Ch(q)| = o\left( \frac{ \log m}{\varepsilon^2} \right)$, is addressed directly via exhaustive search, see proof of Theorem~\ref{thm:main} for details. 


\begin{align}
 \max_{\alpha_x, \sigma_c} & \ \ \ \ \ 0 \nonumber \\
  \textrm{subject to}  
& \ \ \sum_{x \in U_{q,c}(y) } \alpha_x = 1 \qquad \forall c \in \Ch(q) \nonumber \\
   & \ \ \sigma_c = \sum_{x \in U_{q,c}(y) } \alpha_x \ x \qquad \forall c \in \Ch(q) \nonumber  \\
   & \ y^T A_{q,p} z + \sum_{c \in \Ch(q)} y^T A_{q,c} \ \sigma_c \geq e_j^T A_{q,p} z + \sum_{c \in \Ch(q)} e_j^T A_{q,c} \  \sigma_c - \frac{\varepsilon}{2} \quad \forall j \in [m] \label{ineq:brp} \\
   & \ \ \alpha_x \geq 0 \qquad  \forall x \in \cup_{c \in \Ch(q)} \ U_{q, c}(y)  \nonumber \\ 
   & \ \ \sigma_c \in \Delta^m \qquad \  \forall c \in \Ch(q). \nonumber
    \end{align}   

Formally, Lemma~\ref{lemma:relax} below establishes that the feasibility of the linear program LP($p,q,z,y$) implies the required containment $y \in U_{p,q}(z)$, when $|\Ch(q)| = \Omega \left( \frac{ \log m}{\varepsilon^2} \right) $. Note that LP($p,q,z,y$) is parameterized by players $p$ and $q$ along with mixed strategies $z$ and $y$. In addition, inequality (3) in LP($p,q,z,y$) enforces that $y$ is an $\varepsilon/2$-best response against $\sigma_c$s and $z$. Also, if for some player $c \in \Ch(q)$ the set $U_{q,c}(y)$ is empty, then LP($p,q,z,y$) is trivially infeasible.

\begin{lemma}
\label{lemma:relax}
Let player $p$ be the parent of player $q$ in a normalized polymatrix game over rooted tree $G=(V,E)$. Also, let the number of children of $q$, $|\Ch(q)| =  \Omega\left( \frac{\log m}{\varepsilon^2} \right)$. Then, the feasibility of the linear program LP($p,q,z,y$), for mixed strategies $z,y \in U$, implies that $y \in U_{p,q}(z)$. Moreover, using a feasible solution of LP($p,q,z,y$) we can find mixed strategy profiles $E_{p,q}(z,y)$ via a sampling algorithm whose expected running time is polynomial in $|U|$.
\end{lemma}

\begin{proof}
Scalars $(\alpha_x)_{x \in U_{q,c}(y)}$ are nonnegative and sum up to one; hence, they induce a probability distribution over $U_{q,c}(y)$, for $c\in \Ch(q)$. Write $\alpha^c$ to denote this distribution. Also, let $\chi_c$ be the random variable that is equal to mixed strategy $x \in U_{q,c}(y) $ with probability $\alpha_x$, i.e., $\chi_c$ is drawn from $\alpha^c$. Note that $\E_{\alpha^c} [\chi_c ] = \sigma_c$. 

Let $d$ denote the number of children of $q$, $d := |\Ch(q)|$. For fixed $j \in [m]$, we consider function $f_j(\chi_1, \ldots, \chi_d) : = \sum_{c \in \Ch(q)} e^T_j A_{q,c} \chi_c$. The expected value of the function satisfies  $\E_{\alpha^1, \ldots, \alpha^d} [f_j] = \sum_{c \in \Ch(q)} e_j^T A_{q,c} \  \sigma_c$. 

Given that the underlying game is normalized (see Definition~\ref{def:norm}) and $d = \Omega\left(\frac{\log m }{ \varepsilon^2} \right)$, each entry of $A_{q,c}$ is between $0$ and $\frac{\varepsilon}{2\sqrt{6 d \log m}}$. 

This entry-wise bound implies that for any $c \in \Ch(q)$ and $\chi_1,..,\chi_c,.., \chi_d, \chi'_c \in U$ the following Lipscihtz condition holds for $f_j$: 
\begin{align*}
\left|f_j(\chi_1,\ldots, \chi_c, \ldots, \chi_d) -  f_j(\chi_1,\ldots, \chi'_c, \ldots, \chi_d)   \right| & \leq \frac{\varepsilon}{2\sqrt{6 d \log m}} . 
\end{align*}

Using McDiarmid's inequality (see Section~\ref{sec:notation}) we get that 
\begin{align*}
\Pr_{\alpha^1, \ldots, \alpha^d} (|f_j - \E[f_j] | \geq \varepsilon/4 ) &  \leq \frac{2}{m^3}. 
\end{align*}


Say, $\mathcal{E}$ denotes the event that for all $j \in [m]$, we have $|f_j - \E[f_j] | \leq \varepsilon/4$. Using the union bound we get that $\Pr_{\alpha^1,..,\alpha^d} (\mathcal{E}) \geq 1 - 2/m^2$. Therefore, the probabilistic method guarantees the existence of mixed strategies $x_c \in U_{q,c}(y)$, for $c \in \Ch(q)$, that satisfy $\mathcal{E}$. Note that to obtain such a collection of mixed strategies the expected number of times that we need to sample---the product distribution $\prod_{c \in \Ch(q)} \alpha^c$---is at most two. 

Say that mixed strategies $x_c \in U_{q,c}(y)$, for $c \in \Ch(q)$, satisfy event $\mathcal{E}$. Next we will show that $y$ is an $\varepsilon$-best response of $q$ against $(x_c)_{c \in \Ch(q)}$, and $z$. Overall, this implies that $y \in U_{p,q}(z)$, and we can set $E_{p,q}(z,y) = (x_c)_{c \in \Ch(q)}$. 

Mixed strategies $x_c$s satisfy $|f_j(x_1, \ldots, x_d) - \E[f_j] | \leq \varepsilon/4$ for all $j \in [m]$. That is, 
\begin{align}
\label{ineq:bounded}
\left| \sum_{c \in \Ch(d)} e^T_j A_{q,c} \ x_c -  \sum_{c \in \Ch(q)} e_j^T A_{q,c} \  \sigma_c  \right| & \leq \frac{\varepsilon}{4} \qquad \forall j \in [m].
\end{align}
 
Using inequality (\ref{ineq:bounded}) for each $j$ in the support of distribution $y \in \Delta^m$, we have 
\begin{align}
\label{ineq:bndy}
\left| \sum_{c \in \Ch(d)} y^T A_{q,c} \ x_c -  \sum_{c \in \Ch(q)} y^T A_{q,c} \  \sigma_c  \right| & \leq \frac{\varepsilon}{4}.
\end{align}

Note that $y$ satisfies inequality (\ref{ineq:brp}) in the linear program, i.e., $y$ is an $\varepsilon/2$-best response against $\sigma_c$s and $z$. Using inequalities (\ref{ineq:bounded}) and  (\ref{ineq:bndy}) to bound the change in the left-hand-side of (\ref{ineq:brp}) and the right-hand-side of (\ref{ineq:brp}) respectively, we get that $y$ is an $\varepsilon$-best response against $x_c$s and $z$:
\begin{align*}
y^T A_{q,p} z + \sum_{c \in \Ch(q)} y^T A_{q,c} \ x_c & \geq e_j^T A_{q,p} z + \sum_{c \in \Ch(q)} e_j^T A_{q,c} \  x_c - \varepsilon \qquad \forall j \in [m]
\end{align*}

Therefore, if LP($p,q,z,y$) is feasible then $y \in U_{p,q}(z)$. Also, note that the size of LP($p,q,z,y$) is at most $O(nm|U|)$, therefore we can solve the linear program in time polynomial in $|U|$. As mentioned above, given a feasible solution of LP($p,q,z,y$), to obtain $E_{p,q}(y,z)$ (i.e., a collection of mixed strategies $(x_c)_{c \in \Ch(q)}$ that satisfy $\mathcal{E}$)  the expected number of times that we need to sample is at most two. This establishes the running time bound stated in the lemma, and we get the desired claims. 
\end{proof}

Next we prove the main result.
\begin{theorem}
\label{thm:main}
Given an $n$-player $m$-action normalized polymatrix game over a tree, Algorithm~\ref{alg:dp} determines an $\varepsilon$-Nash equilibrium of the game in expected time \[ m^{O\left( \frac{ \log m (\log m + \log n - \log \varepsilon) }{\varepsilon^4} \right)}. \] 
\end{theorem}
\begin{proof}
Let $G=(V,E)$ be the underlying tree of the given normalized polymatrix game. First, we will prove that Algorithm~\ref{alg:dp} necessarily finds a mixed strategy in $U_r$, for the root $r$ of $G$, in the specified amount of time. Hence, via lemma~\ref{lem:backtrack}, we get that Algorithm~\ref{alg:dp} successfully finds an $\varepsilon$-Nash equilibrium of the game. 

As mentioned above, it was established in~\cite{BBP} that every $n$-player $m$-action game admits an $\varepsilon/2$-Nash equilibrium $(\hat{x}_q)_{q \in [n]} $ where each $\hat{x}_q \in U$.\footnote{The change from $\varepsilon$-Nash equilibrium to $\varepsilon/2$-Nash equilibrium can be easily addressed by adjusting the size of $U$.} Hence, for each parent-child pair $(p,q) \in E$ there exists at least one set $U_{p,q}$ which is nonempty; in particular, $ \hat{x}_q \in U_{p,q} (\hat{x}_p)$. Moreover, for $z= \hat{x}_p$ and $y = \hat{x}_q$ the relaxation LP($p,q, z, y$) is guaranteed to be feasible.  Therefore, contingent on the fact that the ``if'' condition in Step~\ref{step:test} and~\ref{step:final} is performed correctly, we get that Algorithm~\ref{alg:dp} is guaranteed to move up the tree with non-empty $U_{p,q}$s and, finally, find a mixed strategy in $U_r$. 

Specifically, the correctness of the ``if'' condition (which we establish below) ensures that for an $\varepsilon/2$-Nash equilibrium $(\hat{x}_p)_{p \in [n]}$ and the sets $U_{p,q}$s populated by the algorithm we have $\hat{x}_q \in U_{p,q}(\hat{x}_p)$ for every parent-child pair $(p,q) \in E$. This follows via an inductive argument over levels of the tree: if $q$ is a leaf node then $\hat{x}_q$ is an $\varepsilon$ best response against $\hat{x}_p$ and we get the desired containment $\hat{x}_q \in U_{p,q}(\hat{x}_p)$. Furthermore, using the induction hypothesis that $\hat{x}_{c} \in U_{q,c}( \hat{x}_q) $ for all $c \in \Ch(q)$, we get that the ``if'' condition in Step~\ref{step:test} will be satisfied for $\hat{x}_q $ and $\hat{x}_p$, i.e., the algorithm will include $\hat{x}_q$ in $U_{p,q} ( \hat{x}_p)$ and the inductive claim holds. In particular, this observation implies that the algorithm will never encounter the situation wherein the set $U_{p,q}(x)$ is remain empty for all $x \in U$ after the for loops, i.e., the algorithm will always run to completion.  It is also relevant to note that the algorithm can set $E_{p,q}(\hat{x}_q, \hat{x}_p)$ to be any tuple $(x_c)_{c \in \Ch(q)}$ that satisfies satisfies the best response condition for $\hat{x}_q$ and $\hat{x}_p$. That is, it is not necessary that the algorithm sets $E_{p,q}(\hat{x}_q, \hat{x}_p) = (\hat{x}_c)_{c \in \Ch(q)}$. But still, the above mentioned argument goes though and we get that the algorithm always runs to completion. 

The ``if'' condition in Step~\ref{step:test} and~\ref{step:final} is performed $O(n|U|^2)$ times. Next we show that the ``if'' condition is verified correctly in expected time $|U|^{O\left(\frac{\log m}{\varepsilon^{2} }\right)} $. This overall establishes the stated claims. 

If the number of children of a player $q$ is $o\left( \frac{\log m}{\varepsilon^2} \right)$ then we can go over the entire set $\prod_{c \in \Ch(q)} U_{q,c}(y)$ in time $|U|^{o\left(\frac{ \log m}{\varepsilon^2} \right)}$ and determine whether the ``if'' condition in Step~\ref{step:test} is satisfied. The same argument works in Step~\ref{step:final}, if the number of children of the root $r$ is $o\left(\frac{ \log m}{\varepsilon^2} \right)$. 

For the remainder of the proof we consider the other case wherein the number of children of $q$ (or the root $r$) is $\Omega\left( \frac{ \log m}{\varepsilon^2} \right)$. In this case we verify the ``if'' condition in Step~\ref{step:test} (and Step~\ref{step:final}) by solving the linear-programming relaxation LP($p,q, z, y$) and employing Lemma~\ref{lemma:relax}. Note the the size of LP($p,q, z, y$) is $O(n|U|)$ and hence (again, via Lemma~\ref{lemma:relax}) in expected time polynomial in $|U|$ we can test if $y \in U_{p,q}(z)$ and find $E_{p,q}(z,y)$. Recall that this test is guaranteed to succeed for the $\varepsilon/2$-Nash equilibrium $(\hat{x}_q)_{q \in [n]} $, since the corresponding  LP($p,q, z, y$)s will be feasible. Hence, we get that Algorithm~\ref{alg:dp} proceeds up the tree with $\hat{x}_q$s, and eventually after processing the root $r$ finds an $\varepsilon$-Nash equilibrium of the game.

Steps~\ref{step:test} and~\ref{step:final} are executed $O(n|U|^2)$ times, and the expected running time of these steps is $|U|^{O\left(\frac{\log m}{\varepsilon^2} \right)}$. These observations establish the time complexity of the algorithm and complete the proof.  
\end{proof}

\section*{Acknowledgements}
This work was supported by NSF grants CNS-0846025, CCF-1101470, CNS-1254169, CNS-1518941, SUTD grant SRG ESD 2015 097, along with a Microsoft Research Faculty Fellowship, a Google Faculty Research Award, a Linde/ SISL Postdoctoral Fellowship and a CMI Wally Baer and Jeri Weiss postdoctoral fellowship. Katrina Ligett gratefully acknowledges the support of the Charles Lee Powell Foundation. The bulk of this work was conducted while the authors were at Caltech. The authors also collaborated on this paper while visiting the Simons Institute for the Theory of Computing.

\bibliographystyle{plain}
\bibliography{polymatrix}

\end{document}